\documentclass{article}

\usepackage[english]{babel}
\usepackage[utf8x]{inputenc}
\usepackage{amsmath}
\usepackage{amsthm}
\usepackage{amsfonts}
\usepackage{graphicx}
\usepackage[colorinlistoftodos]{todonotes}
\usepackage{fullpage}

\newtheorem{theorem}{Theorem}

\newtheorem{definition}[theorem]{Definition}
\newtheorem{lemma}[theorem]{Lemma}
\newtheorem{corollary}[theorem]{Corollary}
\newtheorem{claim}[theorem]{Claim}

\newcommand{\C}{\mathcal{C}}
\newcommand{\F}{\mathcal{F}}
\newcommand{\cS}{\mathcal{S}}
\newcommand{\U}{\mathcal{U}}

\title{The space complexity of mirror games}
\author{
  Sumegha Garg\thanks{
    Department of Computer Science, Princeton University, Princeton, USA.
    Email: {\tt sumeghag@cs.princeton.edu}.
  }
  \and
  Jon Schneider\thanks{
    Department of Computer Science, Princeton University, Princeton, USA.
    Email: {\tt  js44@cs.princeton.edu}.
  }
 }
 \date{}
\begin{document}

\maketitle
\abstract{We consider a simple streaming game between two players Alice and Bob, which we call the mirror game. In this game, Alice and Bob take turns saying numbers belonging to the set $\{1, 2, \dots,2N\}$. A player loses if they repeat a number that has already been said. Bob, who goes second, has a very simple (and memoryless) strategy to avoid losing: whenever Alice says $x$, respond with $2N+1-x$. The question is: does Alice have a similarly simple strategy to win that avoids remembering all the numbers said by Bob?

The answer is no. We prove a linear lower bound on the space complexity of any deterministic winning strategy of Alice. Interestingly, this follows as a consequence of the Eventown-Oddtown theorem from extremal combinatorics. We additionally demonstrate a randomized strategy for Alice that wins with high probability that requires only $\tilde{O}(\sqrt N)$ space (provided that Alice has access to a random matching on $K_{2N}$).

We also investigate lower bounds for a generalized mirror game where Alice and Bob alternate saying $1$ number and $b$ numbers each turn (respectively).  When $1+b$ is a prime, our linear lower bounds continue to hold, but when $1+b$ is composite, we show that the existence of a $o(N)$ space strategy for Bob implies the existence of exponential-sized matching vector families over $\mathbb{Z}^N_{1+b}$.}
\section{Introduction}

\subsection{The mirror game and mirror strategies}

Consider the following simple game. Alice and Bob take turns saying numbers belonging to the set $\{1, 2, \dots, N\}$. If either player says a number that has previously been said, they lose. Otherwise, after $N$ turns, all the numbers in the set have been spoken aloud, and both players win. Alice says the first number.

If $N$ is even, there is a very simple and computationally efficient strategy that allows Bob to win this game, regardless of Alice's strategy: whenever Alice says $x$, Bob replies with $N+1-x$. This is an example of a \textit{mirror strategy} (and for this reason, we refer to the game above as the \textit{mirror game}). Mirror strategies are an effective tool for figuring out who wins in a variety of combinatorial games (for example, two-pile Nim~\cite{berlekamp2003winning}). More practically, mirror strategies can be applied when playing more complex games, such as chess or go (to varying degrees of success). From a computational perspective, mirror strategies are interesting as they require very limited computational resources - most mirror strategies can be described via a simple transformation of the preceding action.

Returning to the game above, this leads to the following natural question: does Alice have a simple strategy to avoid losing when $N$ is even? Since both players have access to the same set of actions, one may be tempted to believe that the answer is yes - in fact, if $N$ is odd, then Alice can start by saying the number $N$ and then adopt the mirror strategy for Bob  described above for a set of $N-1$ elements. However, when $N$ is even, the mirror strategy as stated does not work. 

To answer the above question, we need to formalize what we mean by simple; after all, Alice has plenty of strategies such as ``say the smallest number which has not yet been said''. One useful metric of simplicity is the metric of \textit{space complexity}, the amount of memory a player requires to implement their strategy (we formalize this in Section \ref{sec:def}). Note that Bob needs only $O(\log N)$ bits of memory to implement his mirror strategy, whereas the naive strategy for Alice (remembering everything) requires $O(N)$ bits. 

In this paper, we show that this gap is necessary; any successful, deterministic strategy for Alice requires at least $\Omega(N)$ bits of memory:

\begin{theorem}\label{thm:intromain}[Restatement of Theorem \ref{thm:main}]
If $N$ is even, then any winning strategy for Alice in the mirror game requires at least $(\log 5 - 2)N - o(N)$ bits of space. 
\end{theorem}

\subsection{Eventown and Oddtown}

While many tools exist in the computer science literature for showing space lower bounds (e.g. communication complexity, information theory, etc.), one interesting feature of this problem absent from many others is that any proof of Theorem \ref{thm:intromain} must depend crucially on the parity of $N$. 

In the study of set families in extremal combinatorics, an ``Oddtown'' is a collection of subsets of $\{1, 2, \dots N\}$ where every subset has even cardinality but each pair of distinct subsets has an intersection of odd cardinality. Likewise, an ``Eventown'' is a collection of subsets of $\{1, 2, \dots, N\}$ where every subset has even cardinality but each pair of distinct subsets has an intersection of even cardinality. In 1969, Berlekamp~\cite{berlekamp1969subsets}, answering a question of Erd\H{o}s, showed that while there exist Eventowns containing up to $2^{N/2}$ subsets, any Oddtown contains at most $N$ subsets. 

It turns out that this exponential gap between the size of Oddtowns and the size of Eventowns is directly responsible for the exponential gap between the space complexity of Alice's strategy and the space complexity of Bob's strategy. (One way to see this connection is that, just as Bob's $O(\log N)$-space strategy involves pairing up the numbers of $\{1, 2, \dots, N\}$, one way to construct an Eventown of size $2^{N/2}$ is to perform a similar pairing, and then consider all subsets formed by unions of pairs). 

The Eventown-Oddtown theorem figures into our proof of Theorem \ref{thm:intromain} in the following way. At a given turn $t$, for each possible state of memory of Alice, label this state with the possible subsets of numbers that could have possibly been said at this time. We show (Lemma \ref{lem:setdifference}) that these subsets must form an Oddtown, or else Bob has some strategy that can force Alice to lose. Since there are a large (exponential) number of total possible subsets (Corollary \ref{cor:rcoveringmax}), and since each Oddtown contains at most $N$ subsets, this implies that Alice's memory must be large.

\subsection{Randomized strategies for Alice}

A natural followup to Theorem \ref{thm:intromain} is whether these lower bounds continue to hold if Alice instead uses a randomized strategy, which only needs to succeed with high probability. 

We provide some evidence to show that this might not be the case. We demonstrate an $O(\sqrt{N}\log^2 N)$-space algorithm for Alice that succeeds with probability $1 - O(1/N)$, as long as Alice is provided access to a uniformly chosen perfect matching on the complete graph $K_{N}$ (Theorem \ref{thm:rand-sqrt}). In addition, even with $O(\log N)$-space (and access to a uniformly chosen perfect matching on $K_{N}$), Alice can guarantee success with probability $\Omega(1/N)$. 

In both of these algorithms, Alice attempts the mirror strategy of Bob, hoping that Bob does not choose the number with no match. In the $O(\sqrt{N}\log^2 N)$ algorithm, Alice decreases her probability of failure by maintaining a set of $\tilde{O}(\sqrt{N})$ possible ``backup'' points she can switch to if Bob identifies the unmatched point. This lets Alice survive until turn $N - \tilde{O}(\sqrt{N})$, whereupon Alice can reconstruct the remaining $O(\sqrt{N})$ elements by maintaining $O(\sqrt{N})$ power sums during the computation.

Since Alice cannot store a perfect matching on $K_N$ in $o(N)$ space, this unfortunately does not give any $o(N)$-space strategies for Alice. In addition, our lower bound techniques from Theorem \ref{thm:intromain} fail to give non-trivial guarantees in the randomized model. Demonstrating non-trivial algorithms or non-trivial lower bounds for the general case is an interesting open problem.

\subsection{General mirror games}

It is possible to generalize the mirror game presented earlier to the case where in each turn Alice says $a$ numbers and Bob says $b$ numbers. We refer to this game as the \textit{$(a,b)$-mirror game}.

By a generalization of the Oddtown theorem that works modulo prime $p$, our proof of Theorem \ref{thm:intromain} immediately carries over to show that if $N$ is not divisible by $p$, then any winning strategy for Bob in the $(1, p-1)$-mirror game requires $\Omega(N)$ space (Theorem \ref{thm:pmain}). 

Interestingly, the natural generalization of the Oddtown theorem is known not to hold modulo composite $m$. Grolmusz showed~\cite{grolmusz2000superpolynomial} that if $m$ is composite, there exists a set family of quasipolynomial size such that every set has cardinality divisible by $m$, but the intersection of any two distinct sets has cardinality not divisible by $m$. The best known upper bounds for the size of such a set family are of the form $2^{O(N)}$, and improving these to any bound of the form $2^{o(N)}$ is an important open problem, with applications in coding theory to the construction of matching vector families \cite{dvir2011matching, bhowmick2014new}.

As long as the size of a modulo $m$ Oddtown is bounded above by $2^{o(N)}$, the proof of Theorem \ref{thm:pmain} still achieves an $\Omega(N)$ lower bound on the space complexity of Bob's strategy. It follows that finding any $o(N)$-space winning strategy for Bob in the $(1, m-1)$ mirror game implies the existence of a $2^{\Omega(N)}$-size modulo $m$ Oddtown, and hence similarly sized matching vector families (Theorem \ref{thm:mv}). Admittedly, Bob may not have an $o(N)$-space winning strategy in this game for unrelated reasons; it is interesting whether it is possible to show a converse result, constructing a low-space strategy for Bob from any $2^{\Omega(N)}$-size modulo $m$ Oddtown.

For other $(a, b)$-mirror games, we understand much less about the complexity of winning strategies (if $a+b$ is prime, then depending on the residue of $N$ modulo $a+b$ it is occasionally possible to show an $\Omega(N)$ lower bound on the lower bound of a winning strategy).

\paragraph{Acknowledgements} We would like to thank Mark Braverman, Sivakanth Gopi, and Jieming Mao for helpful advice and discussions.
\section{Definitions and Preliminaries}\label{sec:def}

\paragraph*{Notation} We write $[N]$ to indicate the set $\{1, 2, \dots, N\}$. 

\subsection{The mirror game} 

The \textit{mirror game on $N$ elements} is a game between two players, Alice and Bob. Alice and Bob take turns (Alice going first) saying a number in $[N]$. If a player says a number that has previously been said (either by the other player or by themselves) they lose and the other player wins. In addition, after $N$ successful turns (so no numbers in $[N]$ remain unsaid), both players are declared winners. 

With unbounded memory, it is clear that both players can easily win this game. We therefore restrict our attention to strategies with bounded memory for an introduction to space-bounded complexity). A \textit{strategy computable in memory $m$} for Alice in the mirror game is defined by an initial memory state $s_0 \in [2^m]$ and a pair of transition functions $f: [N] \times [2^{m}] \times [N] \rightarrow [2^{m}]$ and $g: [2^{m}] \rightarrow [N]$ computable in $SPACE(m)$ (see \cite{arora2009computational} for an introduction to bounded space complexity). The function $f$ takes in the previous reply $b \in [N]$ of Bob, Alice's current memory state $s \in [2^m]$, and the current turn $t \in [N]$, and returns Alice's new memory state $s' \in [2^m]$; the function $g$ takes in Alice's current memory state $s \in [2^m]$ (after updating it based on Bob's move), and outputs her next move $a$. We say a strategy for Alice is a \textit{winning strategy} if Alice is guaranteed to win regardless of Bob's choice of actions. 

By removing the constraint that the transition function $f$ is computable in $SPACE(m)$, we obtain a larger class of strategies, which we refer to as strategies \textit{weakly computable in memory $m$}. Our lower bounds in Theorems \ref{thm:intromain} and \ref{thm:pmain} continue to hold for this larger class of strategies.

\subsection{Eventown and Oddtown}

In this section we review the known literature on the Eventown-Oddtown problem. Note that while in the introduction we only defined the terms ``Eventown'' and ``Oddtown'', there are actually four different classes of set system depending on the parity of cardinalities of the subsets and the parity of the cardinality of the intersection.

\begin{definition} A collection of subsets $\F\subseteq [N]$ forms an \textit{(Odd, Even)-town} of sets if:
\begin{enumerate}
\item For every $F\in\F$, $|F|\equiv 1 \mod 2$.
\item For every $F_1\neq F_2\in\F$, $|F_1\cap F_2|=0\mod 2$.
\end{enumerate}

We define \textit{(Odd, Odd)-towns}, \textit{(Even, Odd)-towns}, and \textit{(Even, Even)-towns} similarly. 

\end{definition}

Note that there exist (Even, Even)-towns and (Odd, Odd)-towns containing exponentially (in $N$) many sets; one simple construction of an (Even, Even)-town is to partition the ground set into pairs (possibly with a leftover element), and consider all sets formed by taking unions of these pairs. In contrast, (Even, Odd)-towns and (Odd, Even)-towns each contain at most $N$ sets.

\begin{lemma}\label{lem:oddtown}
Any (Odd, Even)-town $\F$ has size at most $N$ i.e. $|\F|\le N$. ~\cite{berlekamp1969subsets,oddtown}
\end{lemma}
For completeness, we give the proof below.
\begin{proof}
Let $\F$ be $\{F_1,F_2,...,F_k\}$. Let $v_i\in\{0,1\}^n$ be the characteristic function of $F_i$ i.e. $v_i[j]=1\iff j\in F_i$. We know that $v_i^Tv_j=1$ for $i=j$ and 0 otherwise. We claim that $v_1,v_2,...,v_k$ are linearly independent over $\mathbb{F}_2$ and hence, $k\le N$. We prove this claim by contradiction. If they are not linearly independent, then $\exists \lambda_1,\lambda_2,...,\lambda_k$ not all zero such that $\sum_{i=1}^k\lambda_iv_i=0$. Without loss of generality, assume $\lambda_1\neq 0$. As $v_1^T(\sum_{i=1}^k\lambda_iv_i) = \lambda_1$ over $\mathbb{F}_2$, this implies $\lambda_1=0$; a contradiction. 

\end{proof}
\begin{corollary}\label{cor:even-odd}
Any (Even, Odd)-town contains at most $N$ sets.
\end{corollary}
\begin{proof}
The proof of this corollary with $N+1$ in place of $N$ follows from Lemma \ref{lem:oddtown} by constructing an (Odd, Even)-town in space $[N+1]$ from sets $B_1\cup \{N+1\},B_2\cup \{N+1\},...,B_m\cup \{N+1\}$. For the proof with $N$, refer to \cite{berlekamp1969subsets,oddtown}. \end{proof}

By adapting the above linear algebraic arguments to the field $\mathbb{F}_p$, it is possible to show similar upper bounds on the size of set families with cardinality constraints modulo $p$. We will use the following lemma, due to Frankl and Wilson.

\begin{lemma} \label{lem:genoddtown} ~\cite{frankl1981intersection}
Let $p$ be a prime and $L$ be a set of $s$ integers. Let $B_1,B_2,...,B_m\subseteq [N]$ be a family of subsets such that:
\begin{enumerate}
\item $|B_i| \mod p \notin L$ for all $i\in[m]$.
\item $|B_i\cap B_j| \mod p \in L$ for all $i\neq j$.
\end{enumerate}
Then \[m\le \sum_{i=0}^s{N\choose s}\]
We call such family of subsets a \textit{$(p,L)$-Modtown}.
\end{lemma}

Interestingly, by a result of Grolmusz, there is no straightforward generalization of these results modulo a composite number $k$. 

\begin{theorem} \label{thm:grolmusz}
Let $k$ be a positive integer with $r > 1$ different prime divisors. Then there exists a $c > 0$, such that for every $N$, there exists a family of subsets $B_1, B_2, \dots, B_m$ such that:
\begin{enumerate}
\item $m \geq \exp(c(\log N)^{r}/(\log \log N)^{r-1})$,
\item $|B_i| \equiv 0 \bmod k$ for all $i$,
\item $|B_i \cap B_j| \not\equiv 0 \bmod k$ for all $i \neq j$.
\end{enumerate}
\end{theorem}
\begin{proof}
See \cite{grolmusz2000superpolynomial}.
\end{proof}

This type of set family is captured in coding theory by the definition of a matching vector family.

\begin{definition} 
A \textit{matching vector family} ~\cite{bhowmick2014new} over $\mathbb{Z}^n_m$ of size $t$ is a pair of ordered lists $U=\{u_1,u_2,...,u_t\}$ and $V=\{v_1,v_2,...,v_t\}$ where $u_i,v_j\in \mathbb{Z}^n_m$ such that for all $i$, $u_i^Tv_i=0$, and for all $i\neq j$, $u_i^Tv_j \neq 0$.
\end{definition}

By taking $u_i=v_i$ to be the characteristic vector of $B_i$ above, it is clear that a family of subsets of size $m$ in Theorem \ref{thm:grolmusz} gives rise to a Matching Vector family of size $m$. Matching vector families have deep applications to many problems in coding theory, such as private information retrieval and locally decodable codes \cite{dvir2016, efremenko2012}. Understanding the maximum possible size of a matching vector family is an important open problem in coding theory.

\section{Alice requires linear space}
In this section we prove that Alice requires linear space to win the mirror game when $N$ is even. To do this, we will show that if Alice is at a specific memory state, then the possible sets of numbers that have been said at that point in time must form an (Even, Odd)-town (and hence there are at most $N$ such sets for any memory state). 

Before we get into the main proof, we will define and lower-bound the size of what we call a ``covering collection'' of subsets (this will later allow us to lower bound the total number of possible sets of numbers said by round $t$). 

\begin{definition}
A collection of subsets $\C$ of $[N]$ is \textit{$(p,r)$-covering} if:

\begin{enumerate}
\item each $S \in \C$ has $|S| = pr$.
\item for every $T \subset [N]$ with $|T| = r$, there exists an $S \in \C$ with $T \subset S$. 
\end{enumerate}
\end{definition}

\begin{lemma}\label{lem:rcovering}
Every $(p,r)$-covering collection $\C$ has size at least $\frac{\binom{N}{r}}{\binom{pr}{r}}$.
\end{lemma}
\begin{proof}
Every set in this collection contains at most $\binom{pr}{r}$ sets $T$ with cardinality $r$. There are $\binom{N}{r}$ possible sets $T$.
\end{proof}

When $p=2$, it turns out that the lower bound in Lemma \ref{lem:rcovering} is maximized when $r = N/5$.

\begin{corollary}\label{cor:rcoveringmax}
When $r = N/5$, every $(2,r)$-covering collection has size at least $2^{(\log 5 - 2)N - o(N)}$. 
\end{corollary}

We now prove our main theorem.

\begin{theorem}\label{thm:main}
If $N$ is even, then any winning strategy in the mirror game for Alice requires at least $(\log 5 - 2)N - o(N)$ bits of space. 
\end{theorem}
\begin{proof}
Fix a winning strategy for Alice. Assume this strategy uses $m$ bits of memory, and thus has $M = 2^m$ distinct states of memory.

Call a subset $S$ of $[N]$ \textit{$r$-occurring} if it is possible that immediately after turn $2r$, the set of numbers that have been said is equal to $S$. Let $\cS_r$ be the collection of all $r$-occurring sets. Before diving into the main proof, for any fixed deterministic strategy of Alice, we prove a lower bound on $\cS_r$ i.e.  the number of different subsets of numbers that could have been said in the first $2r$ turns over various strategies of Bob.

\begin{lemma}
$\cS_r$ is $(2,r)$-covering.
\end{lemma}
\begin{proof}
Since $2r$ numbers have been said immediately after turn $2r$, every set in $\cS_r$ has cardinality $2r$. We must show that for any $T \subset [N]$ with $|T| = r$, that there exists an $S$ in $\cS_r$ with $T \subset S$.

Consider the following strategy for Bob: ``say the smallest number in $T$ which has not yet been said''. Note that if Bob follows this strategy, then the set of numbers said by turn $2r$ must contain the entire set $T$. This set belongs to $\cS_r$, and it follows that $\cS_r$ is $(2,r)$-covering.
\end{proof}

Write $N = 2n$, and fix a value $r \in [n]$. For a memory state $x$ out of the $M$ possible memory states and an $r$-occurring set $S$, label $x$ with $S$ if it is possible that Alice is at memory state $x$ when the set of numbers that have been said is equal to $S$. Each state of memory may be labeled with several or none $r$-occurring sets, but each $r$-occurring set must exist as a label to some state of memory (by definition). Let $\U_{x}$ be the collection of $r$-occurring labels for memory state $x$. We want to upper bound the size of $\U_x$. Following lemma along with (Even, Odd)-town Lemma helps us in doing exactly that.

\begin{lemma}\label{lem:setdifference}
If $S_1$ and $S_2$ belong to $\U_{x}$, then $|S_1 \setminus S_2|$ is odd.
\end{lemma}
\begin{proof}
Let $D = S_1 \Delta S_2$, let $D_1 = S_1 \setminus S_2$, and let $D_2 = S_2 \setminus S_1$. Assume to the contrary that $|D_1|$ is even. Note then that $|D_2|$ is also even (since $|S_1| = |S_2| = 2r$) and that $|D| = 2|D_1|$. We'll consider two possible cases for the state of the game after turn $2r$: 1. Alice is at state $x$, and the set of numbers that have been said is $S_1$, and 2. Alice is at state $x$, and the set of numbers that have been said is $S_2$.

Consider the following strategy that Bob can play in either of these cases: ``say the smallest number that has not been said that is not in $D$''. We claim that if Bob uses this strategy, Alice will be the first person (after turn $2r$) to say an element of $D$. Note that Bob will not say an element of $D$ until turn $2n - |D|/2$; this is since:

\begin{enumerate}
\item
If we are in case 1, then all of the numbers in $D_1$ have been said but none of the numbers in $D_2$ have been said. There are therefore $|D_2| = |D|/2$ numbers in $D$ which have not been said, so Bob can avoid saying an element of $D$ until turn $2n - |D|/2$. 

\item
Likewise, if we are in case 2, the argument proceeds symmetrically.
\end{enumerate}

On the other hand, if no element of $D$ is spoken by either player between turn $2r$ and turn $2n - |D|/2$, then at turn $2n - |D|/2$, the only remaining elements belong to $D$ (the set of remaining elements is either $D_1$ or $D_2$, depending on which case we are in). If $|D_1| = |D|/2$ is even, then it is Alice's turn to speak at turn $2n - |D|/2$. It follows that Alice will be the first person after turn $r$ to say an element of $D$.

Let $y_1$ be the memory state of Alice when she first speaks an element of $D$ in case 1, and define $y_2$ similarly. We claim that $y_1 = y_2$. Indeed, since Alice's strategy is deterministic and starts from $x$ in both cases 1 and 2, Bob's strategy plays identically in both case 1 and case 2 until an element of $D$ has been spoken. It follows that Alice must speak the same element of $D$ in both cases. But if this element is in $D_1$, and they are in case 1, then this element has already been said before; similarly, if this element is in $D_2$, and they are in case 2, then this element has also been said before. Regardless of which element in $D$ Alice speaks at this state, there is some case where she loses, which contradicts the fact that Alice's strategy is successful. It follows that $|D_1|$ must be odd, as desired.
\end{proof}

\begin{claim}\label{claim:umax}
$|\U_{x}| \leq N$
\end{claim}
\begin{proof}
We claim the sets in $\U_x$ form an (Even, Odd)-town, from which this conclusion follows (Corollary \ref{cor:even-odd}). Each set in $\U_{x}$ has cardinality $2r$, so all sets have even cardinality. By Lemma \ref{lem:setdifference}, any pair of distinct sets $S_1, S_2 \in \U_x$ has odd $|S_1 \setminus S_2|$. Note that since $|S_1 \cap S_2| = |S_1| - |S_1 \setminus S_2|$, and since $|S_1|$ is even, it follows that $|S_1 \cap S_2|$ is odd, so any pair of sets have an odd cardinality intersection.
\end{proof}

Choose $r = N/5$. By Corollary \ref{cor:rcoveringmax}, $\cS_r$ must have cardinality at least $2^{(\log 5 - 2)N - o(N)}$. Since each element in $\cS_r$ belongs to at least one collection $\U_{x}$, and since each collection $\U_{x}$ has cardinality at most $N$ (Claim \ref{claim:umax}), the number $M$ of memory states $x$ is at least $2^{(\log 5 - 2)N + o(N)}/(N+1) = 2^{(\log 5 - 2)N - o(N)}$. It follows that $m = \log M \geq (\log 5 - 2)N - o(N)$, as desired. This proves Theorem \ref{thm:main}.
\end{proof}

\section{Randomized strategies for Alice}\label{sect:ub}

%!TEX root = main.tex

In this section, we consider randomized strategies for Alice. A \textit{randomized strategy computable in memory $m$} is defined similarly as in Section \ref{sec:def}, with the exception that the transition function no longer needs to be deterministic and need only be computable in $RSPACE(m)$; i.e. it must be computable by a space $m$ Turing machine with access to a $poly(N, m)$-sized read-once random tape. We say a randomized strategy wins with probability $p$ if Alice wins with probability at least $p$ against every possible strategy for Bob. 

Unfortunately, we do not know any randomized strategies with sublinear memory which win with high probability. We therefore relax the definition of randomized strategy above and consider randomized strategies where Alice has oracle access to a perfect matching $M$ chosen uniformly at random from all perfect matchings on $K_{N}$ (we assume here that $N$ is even). Calling this oracle with $x \in [N]$ returns $x$'s match $M(x)$ in this matching, and $M$ may be called arbitrarily many times during the computation of $f$. 

One way to interpret the following upper bounds is as a source of difficulty for proving strong lower bounds for randomized strategies (of the form ``you need linear space to succeed with high probability''), since a randomized strategy with memory $m$ with access to a random matching can be viewed as a convex combination of deterministic strategies weakly computable in memory $m$. This means that any attempt to prove strong lower bounds against randomized strategies must fail against this stronger class of randomized strategies. 

Another way to interpret this model of computation is as a specific case of the setting where Alice has arbitrary read access to her random tape (as opposed to read-once access). It is known \cite{nisan1990read} that having arbitrary read access to randomness is more powerful than read-once access as long as certain probabilistic space classes do not collapse. Proving a strong lower bound for randomized strategies would provide more evidence for this separation (this time in the setting of streaming games). 

We begin by showing that with only logarithmic space (and access to a random perfect matching), Alice can already win with probability $\Omega(1/N)$. In contrast, without access to a random matching, we know no logarithmic space strategy that succeeds with better than an exponentially small probability.

\begin{theorem}\label{thm:rand-log}
When $N$ is even, there exists a randomized strategy (with access to a uniformly random perfect matching $M$ on $K_{N}$) for Alice with space complexity $O(\log N)$ which succeeds with probability $\Omega(1/N)$.
\end{theorem}
\begin{proof}
Consider the following strategy for Alice. She begins by sampling a uniform element $x$ from $[N]$. On her first turn, Alice says $x$. On subsequent turns, if Bob has just previously said $y$, Alice replies with $M(y)$.

Note that with this strategy, Alice wins playing against Bob if the last number said by Bob on turn $N$ is $M(x)$. This follows since $M$ is a matching, so there is no other number $y \neq M(x)$ that Bob can say where $M(y)$ has also already been said. It follows that if Bob says $M(x)$ at turn $N$, then Alice is guaranteed to win.

What is the probability Bob says $M(x)$ before turn $N$? We will show it is less than $1 - 1/N$. To do this, we first claim that any winning strategy for Bob (i.e. any strategy that doesn't repeat previously said elements) has the same probability of saying $M(x)$ before turn $N$. This follows from symmetry: since a uniformly random perfect matching conditioned on containing a submatching is still a uniformly random perfect matching on the remaining vertices, at any point in the protocol, if $M(x)$ has not been said yet, it has an equally likely chance of being any of the unsaid elements. Therefore, without loss of generality, assume Bob is playing according to the strategy where each turn he says the smallest number that has not been said so far.

Now, note that if $M(x) = N$, Bob will not say $M(x)$ before turn $N$ (there will always be a smaller unsaid element). But this happens with probability $1/N$, and therefore Alice succeeds with probability at least $1/N$.
\end{proof}

We will next show how to extend this idea to construct an $\tilde{O}(\sqrt{N})$ strategy for Alice which succeeds with high probability. To do this, we will need the following folklore result on determining missing elements from a set in a streaming setting.

\begin{definition}
The ``missing $k$ numbers problem'' is a streaming problem where a subset $S$ of cardinality $N-k$ is chosen from $[N]$, and Alice is shown the elements of $S$ one at a time, in some order. Alice's goal is to output the set $[N] \setminus S$. 
 \end{definition}

\begin{lemma}\label{lem:missing-elements}
There exists an $O(k\log N)$-space deterministic algorithm for the missing $k$ numbers problem.
\end{lemma}
\begin{proof}
See~\cite{muthukrishnan2005data}. For completeness, we include the proof here.

Choose a prime $q \in (N, 2N]$. We will perform all subsequent computations over the field $\mathbb{F}_q$. For a subset $T$ of $[N]$, define $p_i(T) = \sum_{x \in T} x^{i}$. We first claim that if $|T| \leq k$, then Alice can recover the elements of $T$ given the values of $p_i(T)$ for $1 \leq i \leq k$. 

To show this, define $e_i(T)$ to be the value of $i$th elementary symmetric polynomial over the elements of $T$ (that is, $e_i(T) = \sum_{T' \subseteq T, |T'| = i}\prod_{x \in T'} x$). Newton's identities allow us to compute (in space $O(\log N)$) the values of $e_i(T)$ for $1 \leq i \leq k$ from the values of $p_i(T)$ for $1 \leq i \leq k$ (since $k < N < q$, all of these operations are valid over $\mathbb{F}_q$). Now, each element $x$ of $T$ is a root of the polynomial

$$x^{k} - e_1(T)x^{k-1} + e_2(T)x^{k-2} - \dots + (-1)^{k}e_k(T) = 0.$$  

We can factor this in space $O(\log N)$ by trying all possible $x \in [N]$, and therefore recover the original set $T$.

To solve the missing $k$ numbers problem, Alice first computes $p_i(S)$ for each $1 \leq i \leq k$ using $O(k\log N)$ space, updating each count every time she receives a new element from $S$. This allows her to compute $p_i([N] \setminus S)$ via $p_i([N]\setminus S) = p_i([N]) - p_i(S)$. Finally, since $|[N]\setminus S| = k$, she can use the previous algorithm to recover the set $[N]\setminus S$.
\end{proof}

\begin{theorem}\label{thm:rand-sqrt}
When $N$ is even, there exists a randomized strategy (with access to a uniformly random perfect matching $M$ on $K_{N}$) for Alice with space complexity $O(\sqrt{N}\log^2 N)$ which succeeds with probability $1 - O(1/N)$.
\end{theorem}
\begin{proof}
Consider the following strategy for Alice. Alice begins by sampling $r = \sqrt{N}$ distinct elements $X = \{x_1, x_2, \dots, x_r\}$ uniformly at random from all $\binom{N}{r}$ subsets of $r$ elements of $N$ (and stores them in memory). 

At every point in the game, Alice uses her $O(\sqrt{N}\log^2 N) = O(r\log^2 N)$ bits of memory to keep track of: i) which elements of $X$ have already been said, and ii) the first $k = r\log N$ power sums of the elements said so far, modulo some prime $q \in (N, 2N]$ (as per the proof of Lemma \ref{lem:missing-elements}).

Alice begins by choosing a random element of $X$ and saying it. Her algorithm in later rounds is defined as follows:

\begin{enumerate}
\item If less than or equal to $k$ elements remain unchosen (i.e., this is turn $N-k$ or later), Alice uses the $O(k\log n)$-space algorithm from Lemma \ref{lem:missing-elements} to compute the remaining unsaid elements. She then says one of these elements, chosen arbitrarily.
\item Otherwise, if Bob has just said $y$, Alice checks if $M(y)$ belongs to $X$. 
\begin{itemize}
\item If it does not belong to $X$, Alice says $M(y)$. 
\item If $M(y)$ belongs to $X$ but has not been said yet, Alice says $M(y)$ and marks down $M(y)$ as having been said.
\item If $M(y)$ belongs to $X$ and has already been said, Alice chooses an element of $X$ uniformly at random from the elements of $X$ which have not yet been said (if there are no such elements, Alice gives up and says a random element). She then marks this element as having been said.
\end{itemize}
\end{enumerate}

Let $T$ be the random variable denoting the first turn in which all of the elements of $X$ have been said. We first claim that if $T \geq N-k$, then Alice succeeds. Indeed, since the algorithm of Lemma \ref{lem:missing-elements} always succeeds, Alice is guaranteed to succeed if she makes it to turn $N-k$. On the other hand, the only way for Alice to fail before turn $N-k$ is if all of the elements of $X$ have already been said, and Alice is forced to say a random element.

We now claim that with high probability, $T \geq N-k$. To see this, note that, similarly as in the proof of Theorem \ref{thm:rand-log}, any non-trivial strategy of Bob will give rise to the same distribution over $T$. This follows from the fact that at any point in the protocol, each unsaid element has the same probability of belonging to the set $X$. 

We can therefore assume that Bob is playing under the strategy where on his turn, he says the smallest number that has not been said so far. Note that under this strategy, Bob will only say numbers less than or equal to $i$ on turn $i$; it follows that if there exists any $x_i$ such that $x_i \geq N-k$, then also $T \geq N-k$. 

We can thus compute

\begin{eqnarray*}
\Pr[T < n-k] &\leq & \Pr[\forall i,\; x_i \leq N-k] \\
&= & \frac{\binom{N-k}{r}}{\binom{N}{r}} \\
&\leq & \left(\frac{N-k}{N}\right)^{r} \\
&= & \left(1 - \frac{\log N}{\sqrt{N}}\right)^{\sqrt{N}} \\
&= & O(1/N).
\end{eqnarray*}

It follows that Alice wins with probability at least $1 - O(\frac{1}{N})$, as desired.

\end{proof}

\section{The $(a,b)$-mirror game}
Consider the following generalization of the mirror game, where Alice says $a$ numbers each turn, and Bob says $b$ numbers each turn. We call this new game the \textit{$(a,b)$}-mirror game.

As with the regular $(1,1)$-mirror game, mirror strategies exist for the class of $(1, b)$-mirror games (and similarly, the $(a,1)$-mirror games).
\begin{theorem}
If $N$ is divisible by $b+1$, then Bob has a winning strategy computable in $O(\log N)$ memory for the $(1, b)$-mirror game.
\end{theorem}
\begin{proof}
Divide the $N$ elements of $[N]$ into $N/(b+1)$ consecutive $(b+1)$-tuples. Any time Alice says an element in a $(b+1)$-tuple, Bob says all the remaining elements in that $(b+1)$-tuple.
\end{proof}

Unlike in the $(1, 1)$-mirror game, we cannot rule out the existence of low space winning strategies for other games and other choices of $N$. In this section, we summarize what we know about low-space winning strategies for $(a,b)$-mirror games.

\subsection{$a+b$ is Prime}
We begin by considering the set of $(a,b)$-mirror games where $a+b$ is a prime, $p$. The Frankl-Wilson Lemma (Lemma \ref{lem:genoddtown}) allows us to extend some of our proof techniques from Theorem \ref{thm:main} to this case.

In the case where $a=1$, we have the following analogue to Theorem \ref{thm:main}, fully characterizing the $N$ where Bob has a low space winning strategy and where Bob requires linear space.
\begin{theorem}\label{thm:pmain}
Let $p$ be a prime. If $N$ is not divisible by $p$, then any winning strategy for Bob in the $(1, p-1)$-mirror game requires $\Omega(N)$ space.
\end{theorem}
\begin{proof}
See Appendix \ref{sec:papp}.
\end{proof}

When $a>1$, our proof techniques no longer allow us completely characterize the set of $N$ where Bob requires $\Omega(N)$ space to win. Instead, we only have the following partial characterization.

\begin{theorem}\label{thm:pmainab}
Let $p$ be a prime. If $a+b = p$, and $N \bmod p \in \{ a, a+1, \dots, p-1 \}$, then any winning strategy for Bob in the $(a, b)$-mirror game requires $\Omega(N)$ space.
\end{theorem}
\begin{proof}
See Appendix \ref{sec:papp}.
\end{proof}

\subsection{$a+b$ is Composite}
The failure of the Frankl-Wilson lemma (Lemma \ref{lem:genoddtown}) to hold modulo composite numbers prevents us from directly applying our proof techniques in this case. Gromulz's construction (Theorem \ref{thm:grolmusz}) shows that there exist $(m,L)$-Modtown families containing a superpolynomial in $N$ number of sets when $m$ is composite. 

However, a sufficiently small superpolynomial upper bound on the size of such a family would still allow us to prove the analogue of Theorem \ref{thm:pmain}. In fact, any upper bound of $F$ on the size of such a family leads to a lower bound of $\log(2^{\Omega(n)}/F) = \Omega(n) - \log F$ on Bob's memory. This implies the following ``contrapositive'' to Theorem \ref{thm:pmain}:

\begin{theorem}\label{thm:mv} If Bob has a $o(N)$-space winning strategy in a $(1,m-1)$ game when $N\equiv k\not\equiv 0\mod m$, then there exists $2^{\Omega(N)}$ sized matching vector families over $\mathbb{Z}^N_m$ (for sufficiently large $N$). 
\end{theorem}
\begin{proof}
See Appendix \ref{sec:papp}.
\end{proof}

\section{Open Problems}

Our understanding of the space complexity aspects of mirror games is still very rudimentary. We conclude by listing some open problems we find interesting.

\begin{enumerate}
\item \textbf{When do low-space winning strategies exist?} Does either Alice or Bob ever have a low-space winning strategy for any $(a,b)$ game when $a$ and $b$ are both larger than $1$ (e.g., the $(2,2)$ game)? Are there any cases where the best deterministic winning strategy has space complexity strictly between $O(\log N)$ and $O(N)$ (e.g. $O(\sqrt{N})$), or does the best winning strategy always fit into one of these two extremes? Is it ever the case that both Alice and Bob simultaneously have $O(\log N)$ winning strategies?

\item \textbf{Beating low-space with low-space.} In order to show that Alice needs $\Omega(N)$ space, our adversarial strategy for Bob also requires $\Omega(N)$ space. Given a deterministic low-memory $m$ strategy for Bob, is there a low-memory strategy for Alice which wins against it?

\item \textbf{Lower bounds for randomized strategies.} Can we prove any sort of lower bound against randomized strategies that win with high probability? What about against randomized strategies with additional power, such as access to a uniformly chosen matching or with multiple read access to the random tape? Is our upper bound of $\tilde{O}(\sqrt{N})$ tight in these contexts?

\item \textbf{Composite $a+b$.} Is it possible to show some sort of converse to Theorem \ref{thm:mv}, that any large enough matching vector family gives rise to a low space strategy for Bob? Or is it possible to show linear space lower bounds for composite $a+b$ via some other approach?

\item \textbf{More general mirror games.} One of the great successes of the theory of combinatorial games \cite{berlekamp2003winning} is that its techniques apply to essentially all sequential two-player games with perfect information. Is there a more general class of games beyond the family of $(a,b)$-mirror games with similar space complexity phenomena? One possible other family of such games with mirror strategies comes from a generalization of the combinatorial game ``Cram'', where Alice and Bob take turn placing dominoes (more generally, an element of some set of ``symmetric'' polyominos) onto a $w$ by $h$ grid (more generally, some high dimensional grid, or any ``symmetric'' subset of a high-dimensional grid). For example, the original mirror game can be thought of as Cram on a $1$ by $N$ board, where players take turn placing unit squares. Can we say anything about the space complexity of playing Cram, or its generalizations? %\sumeghanote{How is this similar to mirror game? players can always see the board right?}
\end{enumerate}

\bibliographystyle{alpha}
\bibliography{bibliography}
\appendix
\section{Omitted proofs}\label{sec:papp}

\subsection*{Proof of Theorem \ref{thm:pmain}}

\begin{proof}
We proceed similarly to the proof of Theorem \ref{thm:main}. As before, call a subset $S$ of $[N]$ $r$-occurring if it is possible that immediately after turn $2r$, the set of numbers that have been said is equal to $S$. Let $\cS_r$ be the collection of all $r$-occurring sets.
%% appendix here ?
\begin{lemma}\label{lem:cov}
$\cS_r$ is $(p, r)$-covering. 
\end{lemma}
\begin{proof}
Since $pr$ numbers have been said immediately after turn $2r$, every set in $\cS_r$ has cardinality $pr$. We must show that for any $T \subset [N]$ with $|T| = r$, that there exists an $S$ in $\cS_r$ with $T \subset S$.

Consider the following strategy for Alice: ``say the smallest number in $T$ which has not yet been said''. Note that if Alice follows this strategy, then the set of numbers said by turn $2r$ must contain the entire set $T$. This set belongs to $\cS_r$, and it follows that $\cS_r$ is $(p,r)$-covering.
\end{proof}

Similar to Corollary \ref{cor:rcoveringmax}, we can choose $r$ such that every $(p,r)$-covering collection has size at least $2^{\Omega(n)}$. Here, we think of $p$ as constant. \\

Write $N = pn+k$, $k\neq 0$, and fix a value $r \in [n]$. For a memory state $x$ out of the $M$ possible memory states of Bob and an $r$-occurring set $S$, label $x$ with $S$ if it is possible that Bob is at memory state $x$ when the set of numbers that have been said is equal to $S$. Each state of memory may be labelled with several $r$-occurring sets or none, but each $r$-occurring set must exist as a label to some state of memory (by definition). Let $\U_{x}$ be the collection of $r$-occurring labels for memory state $x$.

\begin{lemma}\label{lem:psetdifference}
If $S_1$ and $S_2$ belong to $\U_x$, then $|S_1 \setminus S_2| \equiv k \bmod p$.
\end{lemma}

\begin{proof}
Let $D = S_1 \Delta S_2$, let $D_1 = S_1 \setminus S_2$, and let $D_2 = S_2 \setminus S_1$. Assume to the contrary that $|D_1|\equiv k_1\not\equiv k\mod p$. Note then that $|D_2|=|D_1|$ (since $|S_1| = |S_2|$) and that $|D| = 2|D_1|$. We'll consider two possible cases for the state of the game after turn $2r$: 1. Bob is at state $x$, and the set of numbers that have been said is $S_1$, and 2. Bob is at state $x$, and the set of numbers that have been said is $S_2$. 

Consider the following strategy that Alice can play in either of these cases: ``say the smallest number that has not been said that is not in $D$''. We claim that if Alice uses this strategy, Bob will be the first person (after turn $2r$) to say an element of $D$. Note that Alice will not say an element of $D$ until turn $2\lfloor\frac{N- |D_1|}{p}\rfloor+1$; this is since:

\begin{enumerate}
\item
If we are in case 1, then all of the numbers in $D_1$ have been said but none of the numbers in $D_2$ have been said. There are therefore $|D_2| = |D_1|$ numbers in $D$ which have not been said, so Alice can avoid saying an element of $D$ until turn $2\lfloor\frac{N- |D_1|}{p}\rfloor+1$ as the number of numbers remaining after Alice's turn would be $|D_1|+|k_1-k|-1\ge |D_1|$ (as $k_1\neq k$). 

\item
Likewise, if we are in case 2, the argument proceeds symmetrically.
\end{enumerate}

On the other hand, if no element of $D$ is spoken by either player between turn $2r$ and turn $2\lfloor\frac{N- |D_1|}{p}\rfloor+1$, then at turn $2\lfloor\frac{N- |D_1|}{p}\rfloor+2$, Bob has to say $p-1$ elements from $|D_1|+|k_1-k|-1$ elements and hence, there would be intersection with elements of $D$. 

Let $y_1$ be the memory state of Bob when he first speaks an element of $D$ in case 1, and define $y_2$ similarly. We claim that $y_1 = y_2$. Indeed, since Bob's strategy is deterministic and starts from $x$ in both cases 1 and 2, Alice's strategy plays identically in both case 1 and case 2 until an element of $D$ has been spoken. It follows that Bob must speak the same element of $D$ in both cases. But if this element is in $D_1$, and they are in case 1, then this element has already been said before; similarly, if this element is in $D_2$, and they are in case 2, then this element has also been said before. Regardless of which element in $D$, Bob speaks at this state, there is some case where he loses, which contradicts the fact that Bob's strategy is successful. It follows that $|D_1|\equiv k\bmod p$, as desired.
\end{proof}

\begin{claim}\label{claim:pumax}
$|\U_{x}| \leq N+1$.
\end{claim}
\begin{proof}
We claim the sets in $\U_x$ form a $(p,\{p-k\})$-Modtown, from which this conclusion follows (Lemma \ref{lem:genoddtown}). Each set in $\U_{x}$ has cardinality $pr\equiv 0\mod p$. By Lemma \ref{lem:psetdifference}, any pair of distinct sets $S_1, S_2 \in \U_x$ has $|S_1 \setminus S_2|\equiv k\mod p$. Note that since $|S_1 \cap S_2| = |S_1| - |S_1 \setminus S_2|$, it follows that $|S_1 \cap S_2|\equiv -k\not\equiv 0\mod p$.\end{proof}

We established that for some $r$, $\cS_r$ must have cardinality at least $2^{\Omega(N)}$. Since each element in $\cS_r$ belongs to at least one collection $\U_{x}$, and since each collection $\U_{x}$ has cardinality at most $N+1$ (Claim \ref{claim:pumax}), the number $M$ of memory states $x$ is at least $2^{\Omega(N)}/(N+1) $. It follows that $m = \log M \geq \Omega(N)$, as desired. This proves Theorem \ref{thm:pmain}.
\end{proof}

\subsection*{Proof of Theorem \ref{thm:pmainab}}

\begin{proof}
We proceed similarly to the proof of Theorem \ref{thm:pmain}. As before, call a subset $S$ of $[N]$ $r$-occurring if it is possible that immediately after turn $2r$, the set of numbers that have been said is equal to $S$. Let $\cS_r$ be the collection of all $r$-occurring sets. As before, $\cS_r$ is $(p, r)$-covering and we can choose $r$ such that $|\cS_r|\ge 2^{\Omega(n)}$.\\

Write $N = pn+k$, $k\in \{a,a+1,...,p-1\}$, and fix a value $r \in [n]$. For a memory state $x$ out of the $M$ possible memory states of Bob and an $r$-occurring set $S$, label $x$ with $S$ if it is possible that Alice is at memory state $x$ when the set of numbers that have been said is equal to $S$. Each state of memory may be labelled with several $r$-occurring sets or none, but each $r$-occurring set must exist as a label to some state of memory (by definition). Let $\U_{x}$ be the collection of $r$-occurring labels for memory state $x$.

\begin{lemma}\label{lem:psetdifference}
If $S_1$ and $S_2$ belong to $\U_x$, then $|S_1 \setminus S_2| \bmod p\in\{k-a+1,...,k-1,k,..,p-1\}$.
\end{lemma}

\begin{proof}
Let $D = S_1 \Delta S_2$, let $D_1 = S_1 \setminus S_2$, and let $D_2 = S_2 \setminus S_1$. Assume to the contrary that $|D_1|\equiv k_1\mod p$ where $k_1\in\{0,1,..,k-a\}$ . Note then that $|D_2|=|D_1|$ (since $|S_1| = |S_2|$) and that $|D| = 2|D_1|$. We'll consider two possible cases for the state of the game after turn $2r$: 1. Bob is at state $x$, and the set of numbers that have been said is $S_1$, and 2. Bob is at state $x$, and the set of numbers that have been said is $S_2$. 

Consider the following strategy that Alice can play in either of these cases: ``say the smallest number that has not been said that is not in $D$''. We claim that if Alice uses this strategy, Bob will be the first person (after turn $2r$) to say an element of $D$. Note that Alice will not say an element of $D$ until turn $2\lfloor\frac{N- |D_1|}{p}\rfloor+1$; this is since:

\begin{enumerate}
\item
If we are in case 1, then all of the numbers in $D_1$ have been said but none of the numbers in $D_2$ have been said. There are therefore $|D_2| = |D_1|$ numbers in $D$ which have not been said, so Alice can avoid saying an element of $D$ until turn $2\lfloor\frac{N- |D_1|}{p}\rfloor+1$ as the number of numbers remaining before Alice's turn would be $|D_1|+k-k_1\ge |D_1|+a$. Hence, Alice can say $a$ numbers not belonging to the set $D_1$. 

\item
Likewise, if we are in case 2, the argument proceeds symmetrically.
\end{enumerate}

On the other hand, if no element of $D$ is spoken by either player between turn $2r$ and turn $2\lfloor\frac{N- |D_1|}{p}\rfloor+1$, then at turn $2\lfloor\frac{N- |D_1|}{p}\rfloor+2$, Bob has to say $b$ elements from $|D_1|+k-k_1-a$ elements and hence, there would be intersection with elements of $D$ (as $k-k_1-a\le p-1-a=b-1$). 

Let $y_1$ be the memory state of Bob when he first speaks an element of $D$ in case 1, and define $y_2$ similarly. We claim that $y_1 = y_2$. Indeed, since Bob's strategy is deterministic and starts from $x$ in both cases 1 and 2, Alice's strategy plays identically in both case 1 and case 2 until an element of $D$ has been spoken. It follows that Bob must speak the same element of $D$ in both cases. But if this element is in $D_1$, and they are in case 1, then this element has already been said before; similarly, if this element is in $D_2$, and they are in case 2, then this element has also been said before. Regardless of which element in $D$, Bob speaks at this state, there is some case where he loses, which contradicts the fact that Bob's strategy is successful. It follows that $|D_1|\bmod p\in\{k-a+1,...,k-1,k,..,p-1\}$, as desired.
\end{proof}

\begin{claim}\label{claim:pumax}
$|\U_{x}| \leq p{N\choose p-1}$.
\end{claim}
\begin{proof}
We claim the sets in $\U_x$ form a $(p,\{1,2,...,p-1\})$-Modtown, from which this conclusion follows (Lemma \ref{lem:genoddtown}). Each set in $\U_{x}$ has cardinality $pr\equiv 0\mod p$. By Lemma \ref{lem:psetdifference}, any pair of distinct sets $S_1, S_2 \in \U_x$ has $|S_1 \setminus S_2|\mod p\in\{k-a+1,...,k-1,k,..,p-1\}\subseteq\{1,2,..,p-1\}$ (as $k\ge a$). Note that since $|S_1 \cap S_2| = |S_1| - |S_1 \setminus S_2|$, it follows that $|S_1 \cap S_2|\mod p$ also belongs to $\{1,2,..,p-1\}$.\end{proof}

We established that for some $r$, $\cS_r$ must have cardinality at least $2^{\Omega(N)}$. Since each element in $\cS_r$ belongs to at least one collection $\U_{x}$, and since each collection $\U_{x}$ has cardinality at most $pN^p$ (Claim \ref{claim:pumax}), the number $M$ of memory states $x$ is at least $2^{\Omega(N)}/(pN^p) $. It follows that $m = \log M \geq \Omega(N)$, as desired. This proves Theorem \ref{thm:pmainab}.
\end{proof}

\subsection*{Proof of Theorem \ref{thm:mv}}

\begin{proof}
In proof of Theorem \ref{thm:pmain}, we use the fact that $p$ is prime only to bound the size of $\U_x$ and thus, Lemma \ref{lem:cov} and \ref{lem:psetdifference} hold even for composite $m$. As $\cS_r$ must have cardinality at least $2^{\Omega(N)}$ and there are at most $2^{o(N)}$ memory states $x$, by pigeonhole principle, there exists $x'$ with $|\U_{x'}|\ge 2^{\Omega(N)}$. Let's define the set of vectors $U$ based $U_{x'}$. For every subset $S\in U_{x'}$, we add a vector $v_S$ to $U$ in some order where $v_S$ is the characteristic vector of $S$. Clearly, $|U|=|U_{x'}|$ and $U\subseteq\{0,1\}^N$. We claim the $U$ and $U$ form a Matching Vector family as $v_S^Tv_S=|S|\mod m=0$, for all $S$ and $v_{S_1}^Tv_{S_2}=|S_1\cap S_2|\not\equiv 0\mod m\forall S_1\neq S_2$ by Lemma \ref{lem:psetdifference}. Hence, we have a exponential sized Matching Vector family ($U$, $U$).
\end{proof}

\end{document}